\documentclass[11pt]{article}
\usepackage{}
\usepackage{amsmath,amssymb,amsbsy,amsfonts, amsopn,amstext,amsxtra,euscript,amscd,graphicx}
\usepackage[paper=a4paper, left=3cm, right=3cm, top=3cm, bottom=3cm,
            includefoot, centering]{geometry}
\usepackage{enumerate,float,bm,setspace,boxedminipage,color,rotating}

\newcommand{\remove}[1]{}
\newtheorem{example}{Example}[section]
\newtheorem{definition}{Definition}[section]
\newtheorem{lemma}{Lemma}[section]
\newtheorem{corollary}{Corollary}[section]
\newtheorem{theorem}{Theorem}[section]

\newenvironment{proof}{\vspace{8pt}
\noindent{\bf Proof}: }{{\hfill {\large $\Box$}} \vspace{8pt}}

\begin{document}

\bibliographystyle{plain}

\title{A Coding-Theoretic Application of Baranyai's Theorem}

\author{Liang Feng Zhang\\
Department of Computer Science\\ University of Calgary\\
liangf.zhang@gmail.com}
\date{}
\maketitle

\begin{abstract}
Baranyai's theorem is a well-known theorem in the theory of hypergraphs.
A corollary of this theorem says that one can partition the family of all $u$-subsets
of an $n$-element set into ${n-1\choose u-1}$ sub-families such that each sub-family form
a partition of the $n$-element set, where $n$ is divisible by $u$.
In this paper, we present a coding-theoretic application of Baranyai's theorem
(or equivalently, the corollary). More precisely, we propose the first purely
combinatorial construction of locally decodable codes. Locally decodable codes
are error-correcting codes that allow the recovery of any message bit by looking at
only a few bits of the codeword. Such codes have
attracted a lot of attention in recent years.
We stress that our construction does not improve the parameters of known constructions.
What makes it interesting  is the underlying combinatorial techniques
and their potential in future applications.
\end{abstract}

\section{Introduction}

\subsection{Baranyai's Theorem}
\label{section:baranyai}

Let $\Omega = \{\omega_1, \ldots, \omega_n\}$ be a finite set of cardinality $n$.
A {\em hypergraph} on $\Omega$ is a family
$H =\{E_1, \ldots,E_m\}$ of {\em nonempty} subsets of $\Omega$ such that
$\bigcup_{j=1}^m E_j=\Omega$. The  $n$ elements $\omega_1,\ldots, \omega_n$
 are called {\em vertices} of $H$ and the $m$ subsets $E_1, \ldots,E_m$ are called
  {\em edges} of
$H$.  For every vertex $\omega\in \Omega$, the {\em star} with center $\omega$ is
denoted by $H(\omega)$
and defined to the set of all edges of $H$ that contain $\omega$, i.e.,
$H(\omega)=\{E\in H: \omega\in E\}$.
The {\em degree} of $\omega$ is denoted by $d_H(\omega)$ and defined to be $|H(\omega)|$.
The hypergraph $H$ is said to be {\em $r$-regular} if $d_H(\omega)=r$ for every
$\omega\in \Omega$
and {\em $u$-uniform} if $|E|=u$ for every $E\in H$.
As an alternative, the hypergraph $H$ can also be defined by its {\em incidence matrix}
$A=(a_{ij})_{n\times m}$,   where the rows of $A$ are labeled by the
$n$ vertices $\omega_1,\ldots, \omega_n$,  the columns of $A$ are labeled by  the $m$ edges
$E_1,\ldots, E_m$,   and
\begin{equation}
a_{ij}=\begin{cases}
1,  & {\rm if}~\omega_i\in E_j;\\
0, & {\rm if}~\omega_i\notin E_j
\end{cases}
\end{equation}
for every $i\in [n]\triangleq \{1,2,\ldots,n\}$ and $j\in [m]$.
It is straightforward to see that  every row of the incidence matrix of an $r$-regular
hypergraph contains exactly $r$ 1's and  every column of the incidence matrix of
a $u$-uniform hypergraph contains exactly $u$ 1's.
For every $J\subseteq [m]$, the {\em partial hypergraph} of $H$
generated by $J$ is defined to be the sub-family $\bar{H}=\{E_j: j\in J\}$ of $H$, where
the vertex set of $\bar{H}$ is $\bar{\Omega}=\bigcup_{j\in J} E_j$.
It is also easy to see that the incidence matrix $\bar{A}$ of $\bar{H}$ is in fact
a submatrix
of $A$ whose rows are labeled by elements in $\bar{\Omega}$ and columns are labeled by
elements in $\bar{H}$.  Let $u\in [n]$.  The hypergraph $H$ is called a
{\em $u$-complete
hypergraph of order $n$} and denoted by $K_{n}^u$ if it consists of
all $u$-subsets of $\Omega$.
In particular, we have that $m={n\choose u}$ when $H=K_n^u$.

The  study of complete hypergraphs have been an interesting problem in
the theory of hypergraphs (see Section 5 of Chapter 4 in Berge \cite{Ber89}).
In particular,  Baranyai had an in-depth study of the {\em edge colorings} (see
page 137 of \cite{Ber89} for the definition) of
the complete hypergraphs and obtained  the following theorem (Theorem 11, page 143 of \cite{Ber89}):
\begin{theorem}
{\em (Baranyai 1975)}
Let $n$ and $u$ be integers such that $n \geq u \geq 2$.  Let
 $m_1, \ldots,m_k$  be $k$ positive integers such that $m_1 + \cdots +m_k ={n\choose u}$.
Then the complete hypergraph $K_n^u$ can be divided into
 $k$ partial hypergraphs $H_1, \ldots, H_k$   such that
\begin{itemize}
\item $|H_i|=m_i$ for every $i\in [k]$;
\item $|H_i\cap H_j|=\emptyset$ whenever $i,j\in[k]$ and $i\neq j$; and
\item
$\displaystyle
\left \lfloor \frac{u m_i}{n} \right\rfloor\leq d_{H_i}(\omega)\leq
\left \lceil \frac{u m_i}{n} \right \rceil$ for every $i \in [k]$ and $\omega \in \Omega$.
\end{itemize}
\end{theorem}
In particular, when $u |n$, one can set $k={n-1\choose u-1}$ and $m_1=\cdots=m_k=
n/u$ and
have
the following corollary:
\begin{corollary}
\label{corollary:baranyai}
{\em (Baranyai 1975)}
If $u|n$, then the complete hypergraph $K_n^u$ can be divided into $k={n-1\choose u-1}$
partial hypergraphs $H_1,\ldots, H_k$ such that
\begin{itemize}
\item $H_i$ is a $1$-regular hypergraph
of order $n$ and has $n/u$ edges for every $i\in[k]$; and
\item   $|H_i\cap H_j|=\emptyset$ whenever $i,j\in[k]$ and $i\neq j$.
\end{itemize}
\end{corollary}

In fact, Corollary \ref{corollary:baranyai} says that the family of all $u$-subsets
of an $n$-element set $\Omega$ can be {\em partitioned} into ${n-1\choose u-1}$ sub-families
such that each sub-family forms a {\em partition} of the set $\Omega$.

\begin{example}
\label{example:partition}
Let $n=6$ and $u=2$. The incidence matrix $A$ of $K_n^u$ can be depicted by
Figure \ref{fig:incidence-matrix}, where the rows and columns  of $A$ are labeled by
elements of $\Omega=\{1,2,3,4,5,6\}$ and all 2-subsets of $\Omega$, respectively.
We can divide $K_n^u$  into
5 partial hypergraphs:
$H_1=\{12, 34, 56\},
 H_2=\{13, 25, 46\},
 H_3=\{14, 26, 35\},
 H_4=\{15, 24, 36\}$ and $
 H_5=\{16, 23, 45\}$
 such that both consitions in  Corollary \ref{corollary:baranyai}
 are satisfied.  We also highlight the incidence matrix
 of $H_1$ in Figure \ref{fig:incidence-matrix}.
\begin{figure}[h]
$$\hspace{0.6cm}
  \begin{array}{ccccccccccccccccccccc}
      {\bf 1} & { 1}  & { 1}
    & { 1} & { 1}  & { 2}
    & { 2} & { 2}  & { 2}
    & {\bf 3} & { 3}  & { 3}
    & { 4} & { 4}  & {\bf 5}   \\

      {\bf 2} & { 3}  & { 4}
    & { 5} & { 6}  & { 3}
    & { 4} & { 5}  & { 6}
    & {\bf 4} & { 5}  & { 6}
    & { 5} & { 6}  & {\bf 6}
\end{array}
$$$$
  \begin{array}{c}
   1  \\
  2 \\
  3 \\
  4 \\
  5 \\
  6
  \end{array}	
\left(
  \begin{array}{ccccccccccccccc}
  {\bf 1} & 1 & 1 & 1 & 1 & 0 & 0 & 0 & 0 & {\bf 0} & 0 & 0 & 0 & 0 & {\bf 0} \\
  {\bf 1} & 0 & 0 & 0 & 0 & 1 & 1 & 1 & 1 & {\bf 0} & 0 & 0 & 0 & 0 & {\bf 0} \\
  {\bf 0} & 1 & 0 & 0 & 0 & 1 & 0 & 0 & 0 & {\bf 1} & 1 & 1 & 0 & 0 & {\bf 0} \\
  {\bf 0} & 0 & 1 & 0 & 0 & 0 & 1 & 0 & 0 & {\bf 1} & 0 & 0 & 1 & 1 & {\bf 0} \\
  {\bf 0} & 0 & 0 & 1 & 0 & 0 & 0 & 1 & 0 & {\bf 0} & 1 & 0 & 1 & 0 & {\bf 1} \\
  {\bf 0} & 0 & 0 & 0 & 1 & 0 & 0 & 0 & 1 & {\bf 0} & 0 & 1 & 0 & 1 & {\bf 1} \\
\end{array}
\right)
$$
\label{fig:incidence-matrix}
\caption{Incidence matrix of $K_6^2$}
\end{figure}
\end{example}

Baranyai's theorem has found many interesting applications (for example,
see Section 6, Chapter 4 in \cite{Ber89}).
In this paper, we present a new application of this theorem in the
construction of {\em locally decodable codes} \cite{Yek12}.

\subsection{Locally Decodable Codes}

Let $\mathbb{F}$ be a finite field.
 A classical error-correcting code \cite{MS77} $C: \mathbb{F}^n\rightarrow \mathbb{F}^N$  allows one to encode any
message $x=x_1\cdots x_n$   as a codeword $C(x)$  such that the message  can be recovered
 even if
$C(x)$ gets corrupted in a number of coordinates.
However, to recover even  one symbol
of the message, one has to consider all or most of the coordinates of the
codeword. In such a scenario, more efficient  schemes are possible and they are known as
locally decodable codes.
In such codes, a probabilistic decoder $D$ can  recover any particular symbol $x_i$ of the message
with very good probability by looking at several coordinates of $C(x)$ even
if a constant fraction of $C(x)$ has been corrupted.
 For any $y,z\in \mathbb{F}^N$, we denote by
  $\Delta(y,z)$ their {\em Hamming distance}, i.e., the number of coordinates where they differ.
\begin{definition}
\label{definition:ldc}
{\em (Locally Decodable Code)}
A code $C:\mathbb{F}^n\rightarrow \mathbb{F}^N$ is said to be
$(\rho,\delta,\epsilon)$-{\em locally decodable} if there is a
probabilistic decoder $D$ (which uses  random coins in decoding)  such that
\begin{enumerate}
\item for every $x\in \mathbb{F}^n, i\in[n]$ and $y\in \mathbb{F}^N$ such that
$\Delta(C(x),y)\leq \delta N$, it holds that
$
\Pr[D^{y}(i)=x_i]>1-\epsilon,
$
where the probability is taken over the random coins of  $D$ and
$D^y$ means that $D$ only looks at a number of coordinates of $y$;
\item $D$ looks at at most  $\rho$ coordinates of  the word $y$.
\end{enumerate}
\end{definition}
The quality of $C$ is measured by its  {\em query complexity} $\rho$ and
{\em length} $N$ (both as a function of $n$). Ideally, one would like
 both $\rho$ and $N$ to be as small as possible.

\begin{example}
{\em (Walsh-Hadamard Code, page 249 or 382 of \cite{AB09})} The best example of
locally decodable code  is the
well-known Walsh-Hadamard code $C:\{0,1\}^n\rightarrow \{0,1\}^{2^n}$  whose generator
matrix takes all vectors in $\{0,1\}^n$ as columns. For every message
$x\in \{0,1\}^n$, the coordinates of  $C(x)$  are labeled by the vectors in $\{0,1\}^n$.
In particular, the coordinate labeled by $v\in \{0,1\}^n$ is  equal to
$\sum_{i=1}^n x_iv_i\bmod 2$. Given a word $y\in \{0,1\}^{2^n}$ such that
$\Delta(C(x),y)\leq\delta \cdot 2^n$, a decoder $D$ may recover a  bit $x_i$ by
 looking at two random bits of $y$ labeled by $v, v+e_i\in \{0,1\}^n$ and
then outputs  their sum, where $e_i=(0,\ldots,1,\ldots,0)\in \{0,1\}^n$
is the $i$th unit vector.  Clearly, each of the two bits is corrupted with probability
at most $\delta$ and thus the decoder can output the correct $x_i$ with probability
at least $1-2\delta$. Hence, the Walsh-Haramard code $C$ is a $(2,\delta,2\delta)$-locally
decodable code that encodes $k$-bit messages as $2^n$-bit codewords (i.e.,  $N=2^n$).
\end{example}

Katz and Trevisan \cite{KT00}
were the first to formally define locally decodable codes.
In recent years, the construction of locally decodable codes have attracted a large amount of attention
 \cite{Yek07,Efr09,FZ10}. The interested readers are referred to Yekhanin
 \cite{Yek12} for a good survey of locally decodable codes.

 \subsection{Results}

While the series of works mentioned above focus on improving
the parameters (i.e., $\rho$ and $N$) of the locally decodable codes and require nice algebraic
ideas, in this paper we are
 interested in the {\em connection} between locally decodable codes and  combinatorial objects in
 discrete mathematics. In particular, we propose the {\em first purely combinatorial} construction
 of locally decodable codes  which is  based on the hypergraphs in Section
 \ref{corollary:baranyai}. More precisely, we show the following theorem:
 \begin{theorem}
 \label{theorem:ldc}
For any odd positive integer $\rho$, there is a binary linear
$(\rho, \delta, \rho^2\delta/(\rho-1))$-locally decodable code that
encodes  $n$-bit messages as  $2^{{\bf H}(1/\rho)n}$-bit codewords, where
${\bf H}(s)=-s\log_2 s-(1-s)\log_2(1-s)$ is the binary entropy function. 
\end{theorem}

\section{The  Combinatorial Construction}

In this section, we present our purely combinatorial costruction of locally decodable
codes and prove Theorem  \ref{theorem:ldc}.
We firstly give a  technical lemma and then present both the encoding and decoding algorithms
of our locally decodable codes.

\subsection{A Technical Lemma}

\begin{lemma}
\label{lemma:technical}
 Let $C: \mathbb{F}^n\rightarrow \mathbb{F}^N$
be a linear code with generator matrix $G=[G_1,\ldots, G_N]$, where
$G_1,\ldots, G_N\in \mathbb{F}^n$
are the columns of $G$.
Let $k$ be a positive integer such that $k|\lambda N$, where  $0<\lambda\leq 1$
and $\lambda N$ is an integer.    Suppose there are $n$ subsets $T_1,\ldots, T_n\subseteq [N]$,
each of cardinality $\lambda N$, such that
\begin{itemize}
\item[\rm (i)] for every $i\in [n]$,  the set $T_i$ has a partition
 $T_i=T_{i1}\cup \cdots \cup T_{ik}$ such that
 $|T_{ij}|=\lambda N/k$ for every $j\in [k]$; and
\item[\rm (ii)] for every $i\in [n]$ and $j\in [k]$, the {\em $i$th} unit vector $e_i$ is
a linear combination of   $\{G_\ell: \ell\in T_{ij}\}$.
\end{itemize}
Then the code $C$ is $(\rho,\delta, \rho\delta/\lambda)$-locally decodable, where $\rho=\lambda N/k$.
\end{lemma}

\begin{proof}
We need to provide  a probabilistic decoder $D$ for the code $C$ such that
both requirements in Definition \ref{definition:ldc} are satisfied.
Let $x\in \mathbb{F}^k$  be any message and let $C(x)$ be its codeword.
As required, the decoder $D$ is given access to the coordinates of a
word $y\in \mathbb{F}^N$ and asked to recover a particular symbol of the message
$x$, say $x_i$, where $i\in [n]$.
Given the public knowledge of the subsets $T_1,\ldots, T_n$ and their partitions, our decoder
$D$  picks a random integer $j\in [k]$ and looks at the coordinates of $y$ labeled by
the elements of $T_{ij}$, i.e., $\{y_\ell: \ell\in T_{ij}\}$.  Since the $i$th unit vector
$e_i$ is a linear combination of $\{G_\ell: \ell\in T_{ij}\}$, there are $|T_{ij}|=\lambda N/k$
field elements $\{c_\ell: \ell\in T_{ij} \}$ such that
$e_i=\sum_{\ell\in T_{ij}} c_\ell G_\ell$.  Knowing the field elements
$\{c_\ell: \ell\in T_{ij}\}$, our decoder simply outputs
$\sum_{\ell\in T_{ij}}c_\ell y_\ell$.

 Firstly, it is clear that our decoder looks at
 $\lambda N/k=\rho$ coordinates of $y$, i.e., its query complexity is $\rho$.
 Secondly, we shall  compute $\Pr[D^y(i)=x_i]$.
If $y=C(x)$, then  we  have that
\begin{equation}
\label{equation:recover}
\sum_{\ell\in T_{ij}}c_\ell y_\ell=
 \sum_{\ell\in T_{ij}}c_\ell\cdot  \langle  x, G_\ell\rangle=
 \langle x, \sum_{\ell\in T_{ij}}c_\ell G_\ell \rangle=\langle x, e_i \rangle=x_i,
 \end{equation}
 where $\langle \cdot, \cdot \rangle$ stands for the standard dot product.
Clearly, the equation (\ref{equation:recover}) says that $D$ always succeeds
in recovering $x_i$ when $y=C(x)$, i.e., $\Pr[D^{C(x)}(i)=x_i]=1$.
 Now  suppose that $y\in \mathbb{F}^N$ is a word such that $\Delta(C(x),y)\leq \delta N$
 for a small constant $0\leq \delta < 1$. In this case, the decoder $D$ may not output
 $x_i$ correctly since some of the coordinates $\{y_\ell: \ell\in T_{ij}\}$
may have been corrupted and  consequently the left hand side of
(\ref{equation:recover})  is not equal to $x_i$.  We say that $j\in [k]$ is {\em bad}
if there is at least one $\ell\in T_{ij}$ such that $y_\ell$ is not equal to the $\ell$th coordinate of
$C(x)_\ell$; otherwise,
$j$ is called {\em good}.  Equation (2) shows that $D$ can correctly output
$x_i$ when the
$j$  is good. Therefore, $\Pr[D^y(i)=x_i]$ is at least the probability that
the $j$ chosen by $D$ is good. However, since $\Delta(C(x),y)\leq \delta N$,
at most $\delta N$ of the coordinates of $y_\ell$ are not consistent with $C(x)$ and thus
at most $\lambda N$ of the indices $j\in [k]$ are bad.  It follows that
$$\Pr[D^y(i)=x_i]\geq \Pr[j{\rm~is~good}]\geq 1-\delta N/k=1-\rho \delta/\lambda,$$
which in turn implies that $C$ is a $(\rho,\delta,\rho \delta/\lambda)$-locally decodable code
 that   encodes messages of $n$ symbols as codewords of $N$ symbols.
\end{proof}

\subsection{The Construction}

In this section, we present our combinatorial construction of locally decodable codes.
Our construction is based on the combinatorial objects called  hypergraphs that are defined
in Section \ref{section:baranyai}.

The locally decodable codes we construct here are  linear codes over the binary field
$\mathbb{F}_2$.
Let $\rho>1$ be any odd positive integer  and let $n= \rho u +1$
for any positive integer $u$.  Let $K_n^u$ be the $u$-complete hypergraph
over an $n$-element vertex set $\Omega=\{\omega_1,\ldots, \omega_n\}$ and let
$A=(a_{ij})_{n\times N}=[A_1,\ldots, A_N]$  be the incidence matrix
of $K_n^u$, where $N={n\choose u}$ is the number of all  $u$-subsets of
 $\Omega$ and $A_j$ is the $j$th column of $A$ for every $j\in[N]$.
Since the entries of $A$ are either 0 or 1, we can consider the matrix
$A$ over the binary field $\mathbb{F}_2$. We define $G=(g_{ij})_{n\times N}=[G_1,\ldots, G_N]$
be  a binary matrix  such that $g_{ij}=1+a_{ij}$ for every $i\in[n]$ and $j\in [N]$,
where $G_j$ is the $j$th column for every $j\in[N]$.
Our code   $C:\mathbb{F}_2^n\rightarrow \mathbb{F}_2^N$
has  generator matrix $G$, i.e.,  any message $x\in \mathbb{F}_2$ will be encoded as
$C(x)=x G$.

We have to show that the code $C$ we constructed above is locally decodable.
In fact, this is a consequence of  Corollary \ref{corollary:baranyai} and Lemma
\ref{lemma:technical}. Formally, we have that
\begin{theorem}
The code $C:\mathbb{F}_2^n\rightarrow \mathbb{F}_2^N$  is
$(\rho,\delta,\rho\delta/\lambda)$-locally decodable, where $\lambda=1-u/n.$
\end{theorem}

\begin{proof}
Let $k=\lambda N/\rho={n-1\choose u}/\rho$. Due to Lemma 2, we only need to show that there are
$n$ subsets $T_1,\dots,T_n\subseteq [N]$ such that  both (i) and (ii)  hold.
For every $i\in[n]$, let
$$T_i=\{j\in [N]: g_{ij}=1\}.$$
As the incidence matrix $A$, we can consider the rows and columns of $G$ are labeled by the $n$
elements of $\Omega$ and the $N$ $u$-subsets of $\Omega$, respectively.
Due to the definition of $G$, it is then straightforward to see that $T_i$
corresponds to the set of all $u$-subsets of $\Omega\setminus \{\omega_i\}$ for every
$i\in[n]$.  Clearly, we have that $|T_1|=\cdots=|T_n|={n-1\choose u}=\lambda N$.

We consider the submatrix  $G^{(i)}$ of $G$ that consists of
all columns of $G$ labeled by $T_i$ (or equivalently, by all $u$-subsets of
$\Omega\setminus \{\omega_i\}$).  Let $J$ be the all-one matrix of size
$(n-1)\times {n-1\choose u}$ and
$B$ be the incidence matrix of
the $u$-complete hypergraph of order $n-1$ (i.e., $K_{n-1}^{u}$)
on $\Omega\setminus \{\omega_i\}$.
Clearly, we have that
$$G^{(i)}=
\begin{pmatrix}
F_1\\
{\bf 1}\\
F_2
\end{pmatrix},$$
where ${\bf 1}$ is the all-one row vector of dimension ${n-1\choose u}$,
$F_1$ consists of the first $i-1$ rows of $J-B$ and
$F_2$ consists of the last $n-1-i$ rows of $J-B$.
Since $u|(n-1)$, Corollary \ref{corollary:baranyai} implies that we can  partition the
set of columns of $G^{(i)}$ into $k={n-2\choose u-1}$ subsets, say $T_{i1}, \ldots, T_{ik}\subseteq
T_i$, such that each $T_{ij}$ is a partition of $\Omega\setminus \{\omega_i\}$
for every $j\in[k]$. It is clear that $|T_{ij}|=(n-1)/u=\lambda N/k$ for every $j\in[k]$,
which implies that (i) holds. On the other hand, for every $j\in[k]$, we consider the
submatrix $G^{(i,j)}$ of $G^{(i)}$ that consists of all the columns of $G^{(i)}$
labeled by elements in $T_{ij}$.
Clearly,  there are $\rho$ subsets $S_1,\ldots, S_\rho\subseteq \Omega\setminus\{\omega_i\}$
of cardinality $u$ corresponding to the indices in $T_{ij}$. In particular,
these subsets form a partition of $\Omega\setminus \{\omega_i\}$.
We note that every  $\bar{i}\in [n]\setminus \{i\}$ appears in exactly one of
the $\rho$ sets and therefore the $\bar{i}$th row of $G^{(i,j)}$ contains exactly
$\rho-1$ 1's and one 0's.  On the other hand,  the $i$th row of $G^{(i,j)}$ is the all-one vector
of dimension $\rho$. It follows that the sum of the columns of
 $G^{(i,j)}$ is
$$\begin{pmatrix}
(\rho-1)\cdot J_1\\
1\\
(\rho-1)\cdot J_2
\end{pmatrix},$$
where $J_1$ is the all-one vector of dimension $i-1$ and $J_2$ is the all-one vector of
dimension $n-i$. Clearly, this is the $i$th unit vector over the binary field since $\rho$
is an odd integer.  In other words,  (ii) holds. Due to Lemma \ref{lemma:technical},
our code $C$ must be $(\rho,\delta,\rho\delta/\lambda)$-locally decodable.
\end{proof}

In our construction,  the query complexity $\rho$ of the code $C$ should be a constant
and the quality of the code should
 be measured by the asymptotic length $N$ as a function of $n$.
 Due to basic mathematics, we have that
 $$N={n\choose u}={n\choose (n-1)/\rho}\leq 2^{{\bf H}(1/\rho)n},$$
 where ${\bf H}(s)=-s\log_2 s-(1-s)\log_2(1-s)$ is the binary entropy
 function. It is clear that our code is asymptotically more efficient than
 the Welsh-Hadamard code whenever $\rho>1$ is an odd integer.
For example,  we have that $N\leq 2^{0.92n}$ when $\rho=3$; and
$N\leq 2^{0.44n}$ when $\rho=1$, which is less than the square root of the length
of the Walsh-Hadamard code.
On the other hand, we should also have an estimation of the probability that
our decoder $D$ gives an incorrect output, i.e., $\epsilon=\rho \delta/\lambda$.
However, it is easy to see that $\lambda=1-u/n=1-u/(\rho u+1)>1-1/\rho$ and
therefore $\epsilon<\rho^2\delta/(\rho-1)$, which is a constant as long as $\delta$ is
a constant. Hence, we have the following theorem:

\begin{theorem}
For every odd  integer $\rho>1$,  there is a binary
 linear  $(\rho, \delta, \rho^2\delta/(\rho-1))$-locally decodable code that
encodes  $n$-bit messages as   $2^{{\bf H}(1/\rho)n}$-bit codewords.
\end{theorem}

\begin{example}
Let $n=7$ and $u=2$.  Then generator matrix $G$ of our $(3,
\delta,4.5\delta)$-locally decodable code $C$ can be depicted by
the Figure \ref{fig:generator-matrix}, where the rows and columns of
$G$ are labeled by elements of $\Omega=\{1,2,3,4,5,6,7\}$
and all 2-subsets of $\Omega$, respectively.
For explanation, we also highlight the set $T_7$
and the matrix $G^{(7)}$ in Figure \ref{fig:generator-matrix}.
As an example, in order to recover the 7th message bit,
our decoder $D$ may  look at 3 coordinates
of the received word $y$ that are labeled by any one the following subsets:
$$T_{71}=\{12, 34, 56\},~
 T_{72}=\{13, 25, 46\},~
 T_{73}=\{14, 26, 35\},$$$$
 T_{74}=\{15, 24, 36\},~
 T_{75}=\{16, 23, 45\}.$$
 In fact, these sets correspond to the partition of $K_6^2$ we noted in Example
 \ref{example:partition}.
\end{example}

\begin{figure}[H]
$$\hspace{0.95cm}
  \begin{array}{ccccccccccccccccccccccccccc}
      {\bf 1} &  {\bf 1} &  {\bf 1} &  {\bf 1} &  {\bf 1} &   1 &  {\bf 2} &
      {\bf 2} &  {\bf 2} &  {\bf 2} &        2 &  {\bf 3} &  {\bf 3} &  {\bf 3} &
      3       &  {\bf 4} &  {\bf 4} &  4       &  {\bf 5} &  5       &  6 & \\

      {\bf 2}       &  {\bf 3}       & {\bf  4}       &  {\bf 5}       &  {\bf 6}       &  7       &  {\bf 3} &
      {\bf 4} &  {\bf 5} &  {\bf 6} &  7 &  {\bf 4} &  {\bf 5} &  {\bf 6} &
      7 &  {\bf 5} &  {\bf 6} &  7 &  {\bf 6} &  7 &  7 &
\end{array}
$$$$
  \begin{array}{c}
   1  \\
  2 \\
  3 \\
  4 \\
  5 \\
  6 \\
  7
  \end{array}	
\left(
  \begin{array}{ccccccccccccccccccccc}
  {\bf 0} & {\bf 0} & {\bf 0} & {\bf 0} & {\bf 0} &  0 & {\bf 1} & {\bf 1} & {\bf 1} & {\bf 1} & 1 & {\bf 1} & {\bf 1} & {\bf 1} & 1 & {\bf 1} & {\bf 1} & 1 & {\bf 1} & 1 & 1 \\
  {\bf 0} & {\bf 1} & {\bf 1} & {\bf 1} & {\bf 1} &  1 & {\bf 0} & {\bf 0} & {\bf 0} & {\bf 0} & 0 & {\bf 1} & {\bf 1} & {\bf 1} & 1 & {\bf 1} & {\bf 1} & 1 & {\bf 1} & 1 & 1 \\
  {\bf 1} & {\bf 0} & {\bf 1} & {\bf 1} & {\bf 1} &  1 & {\bf 0} & {\bf 1} & {\bf 1} & {\bf 1} & 1 & {\bf 0} & {\bf 0} & {\bf 0} & 0 & {\bf 1} & {\bf 1} & 1 & {\bf 1} & 1 & 1 \\
  {\bf 1} & {\bf 1} & {\bf 0} & {\bf 1} & {\bf 1} &  1 & {\bf 1} & {\bf 0} & {\bf 1} & {\bf 1} & 1 & {\bf 0} & {\bf 1} & {\bf 1} & 1 & {\bf 0} & {\bf 0} & 0 & {\bf 1} & 1 & 1 \\
  {\bf 1} & {\bf 1} & {\bf 1} & {\bf 0} & {\bf 1} &  1 & {\bf 1} & {\bf 1} & {\bf 0} & {\bf 1} & 1 & {\bf 1} & {\bf 0} & {\bf 1} & 1 & {\bf 0} & {\bf 1} & 1 & {\bf 0} & 0 & 1 \\
  {\bf 1} & {\bf 1} & {\bf 1} & {\bf 1} & {\bf 0} &  1 & {\bf 1} & {\bf 1} & {\bf 1} & {\bf 0} & 1 & {\bf 1} & {\bf 1} & {\bf 0} & 1 & {\bf 1} & {\bf 0} & 1 & {\bf 0} & 1 & 0 \\
  {\bf 1} & {\bf 1} & {\bf 1} & {\bf 1} & {\bf 1} &  0 & {\bf 1} & {\bf 1} & {\bf 1} & {\bf 1} & 0 & {\bf 1} & {\bf 1} & {\bf 1} & 0 & {\bf 1} & {\bf 1} & 0 & {\bf 1} & 0 & 0 \\
\end{array}
\right)
$$
\label{fig:generator-matrix}
\caption{Generator matrix of a 3-query locally decodable code}
\end{figure}

\section{Concluding Remarks}

In this paper, we present an application of Baranyai's theorem.
In particular, we constructed a $(\rho,\delta,\rho^2\delta/(\rho-1))$-locally decodable
code that encodes $n$-bit messages as $2^{{\bf H}(1/\rho)n}$-bit messages, where $\rho>1$
is any odd integer. Our construction does not improve the parameters of the known
constructions
to date. However, it is still interesting in the sense that the underlying
techniques are purely combinatorial while all the known constructions heavily rely
on algebraic techniques.

\bibliography{ldc}

\end{document}